\documentclass[12pt,a4,reqno,twoside]{amsart}
\usepackage{amssymb}
\usepackage{amsthm}
\usepackage{amsmath}
\usepackage{algorithm}
\usepackage{algpseudocode}
\usepackage{graphicx}

\newtheorem{theorem}{Theorem}[section]
\newtheorem{lemma}[theorem]{Lemma}

\newcommand {\coloring}{\chi}

\title{New Hardness Results in Rainbow Connectivity}
\author[P. Ananth]{Prabhanjan Ananth}
\thanks{} 
\address{Dept. of Computer Science and Automation, Indian Institute of
Science}
\email{prabhanjan@csa.iisc.ernet.in}
\author[M. Nasre]{Meghana Nasre}
\thanks{} 
\address{Dept. of Computer Science and Automation, Indian Institute of
Science}
\email{meghana@csa.iisc.ernet.in}
\date{}

\begin{document}

\maketitle

\begin{abstract}
A path in an edge colored graph is said to be a rainbow path 
if no two edges on the path have the same color. An edge colored 
graph is (strongly) rainbow connected if there exists a (geodesic) 
rainbow path between every pair of vertices. The (strong) 
rainbow connectivity of a graph $G$, denoted by ($src(G)$, respectively) 
$rc(G)$ is the smallest number of colors required to edge color 
the graph such that the graph is (strong) rainbow connected. 
It is known that for \emph{even} $k$ to decide whether the 
rainbow connectivity of a graph is at most $k$ or not is NP-hard. 
It was conjectured that for all $k$, to decide whether $rc(G) \leq k$ 
is NP-hard. In this paper we prove this conjecture. We also 
show that it is NP-hard to decide whether $src(G) \leq k$ or 
not even when $G$ is a bipartite graph.
\end{abstract}


\section{Introduction}
This paper deals with the notion of {\em rainbow coloring} and {\em strong rainbow coloring}
of a graph. Consider an edge coloring (not necessarily proper) of a graph $G = (V, E)$. We say that
there exists a {\em rainbow path} between a pair of vertices, if no two edges on the path
have the same color. A coloring of the edges in graph $G$ is called a rainbow coloring if between every
pair of vertices in $G$ there exists a rainbow path. An edge coloring is a {\em strong}
rainbow coloring if between every pair of vertices, one of its geodesic \textit{i.e.,} shortest paths is a rainbow
path. The minimum number of colors required to rainbow color a graph $G$ is called the rainbow
connection number denoted by $rc(G)$.  Similarly, the minimum number of colors required
to strongly rainbow color a graph $G$ is called the strong rainbow connection
number, denoted by $src(G)$. 
The notions of rainbow coloring and strong rainbow coloring were introduced recently
by Chartrand et al.~\cite{chartrand2008rainbow} as a means of strengthening 
the connectivity. Subsequent to this paper, the problem has received attention by
several people and the complexity as well as upper bounds for the rainbow connection number
have been studied.

In this paper we study the complexity of computing $rc(G)$ and $src(G)$. We prove the following results:
\begin{enumerate}
\item For every $k \ge 3$, deciding whether $src(G) \leq k$, is NP-hard even when $G$ is bipartite.
\item Deciding rainbow connection number of a graph is at most 3 is NP-hard even when the graph $G$ is bipartite.
\item For every $k \ge 3$, deciding whether $rc(G) \leq k$ is NP-hard.
\end{enumerate}

We note that Chakraborty et al.~\cite{chakraborty2008hardness} proved that for all even $k$,
deciding whether $rc(G) \leq k$ is NP-hard. For proving the result they introduced
the problem of {\em subset} rainbow connectivity where in addition to the 
graph $G = (V, E)$ we are given a set $P$ containing pairs of vertices. The goal is to answer
whether there exists an edge coloring of $G$ with $k$ colors such that every
pair in $P$ has a rainbow path. We also use the subset rainbow connectivity problem
and analogously define the subset strong rainbow connectivity problem to prove our 
hardness results.

\subsection{Related work}
The concepts of rainbow connectivity and strong rainbow connectivity 
were first introduced by Chartrand et al. in \cite{chartrand2008rainbow}. 
There they computed the rainbow connection number and strong rainbow connection
number for several graphs including the complete bipartite graph and multipartite graphs.
Subsequent to this upper bounds for the rainbow connectivity as a function of 
minimum degree and the number of vertices of the graph were explored 
in \cite{caro2008rainbow},\cite{Krivelevich:2010:RCG:1753069.1753071},\cite{chandran2010rainbow}. 
Graphs of diameter 2 were studied in \cite{li2011rainbow} and it was shown 
that $rc(G)$ is upper bounded by $k+2$ where $k$ is the number of bridges in the graph. 
Some upper bounds were shown in \cite{chandran2010rainbow} 
for special graphs like interval graphs, AT-free graphs. 
The threshold function for random graph to have $rc(G)=2$ was studied in \cite{caro2008rainbow}. 
In \cite{Krivelevich:2010:RCG:1753069.1753071}, the concept of rainbow vertex 
connection number was studied. In \cite{basavaraju2010rainbow}, 
Basavaraju et al. gave a constructive argument to show that any 
graph $G$ can be colored with $r(r+2)$ colors in polynomial time 
where $r$ is the radius of the graph.
In \cite{chakraborty2008hardness}, it was shown that the following 
problem is NP-hard: Given a graph $G$ and an even number $k (>0)$, is $rc(G) \leq k$? 
It was conjectured that it is NP-hard to determine whether $rc(G) \leq k$ for all $k>0$. 
Similar complexity results were shown for the rainbow vertex connection number in \cite{chen2011complexity}.     


\subsection{Organisation of the paper}
We prove complexity results related to strong rainbow connectivity in Section \ref{src}. In Section \ref{rc}, we give the complexity results related to the rainbow connectivity.


\section{Strong rainbow connectivity}
\label{src}
In this section we prove the hardness result for the following problem: given a graph $G$ 
and an integer $k \ge 3$, decide whether $src(G) \leq k$. In order to show the hardness of this problem, we will
first consider an intermediate problem called the $k$-subset strong rainbow connectivity problem which is the decision version of the subset strong rainbow connectivity problem. The input to the $k$-subset strong rainbow connectivity problem is a graph $G$ along with a set of pairs
$P = \{(u, v): (u,v) \subseteq V \times V\}$ and an integer $k$. Our goal is to answer whether there exists
an edge coloring of $G$ with at most $k$ colors such that every pair $(u,v) \in P$ has a geodesic rainbow path.

Our overall plan is to prove that $k$-subset strong rainbow connectivity is NP-hard by showing a reduction
from the vertex coloring problem. We then establish the polynomial time equivalence of the $k$-subset strong rainbow connectivity problem
and the problem of determining whether $src(G) \leq k$ for a graph $G$. 


\subsection{$k$-subset strong rainbow connectivity}
Let $G = (V, E)$ be an instance of the $k$-vertex coloring problem. We say that $G$ can be vertex colored 
using $k$ colors if there exists an assignment of at most $k$ colors to the vertices of $G$
such that no pair of adjacent vertices are colored using the same color. This problem is NP-hard for $k \ge 3$.
Given an instance $G = (V, E)$ of the $k$-vertex coloring problem, we construct
an instance $\langle G' = (V', E'), P\rangle$ of the $k$-subset strong rainbow connectivity problem.

The graph $G'$ that we construct is a star, with one leaf vertex corresponding
to every vertex $v \in V$ and an additional central vertex $a$.
The set of pairs $P$ captures the edges in $E$, that is, for every edge $(u,v) \in E$ 
we have a pair $(u,v)$ in the set $P$. The goal is to color the edges of $G'$
using at most $k$ colors such that every pair in the set $P$ has a geodesic rainbow path.
More formally, we define each of the parameters $\langle G' = (V', E'), P \rangle$ of
the $k$-subset strong rainbow connectivity problem below:
\begin{eqnarray*}
V' = \{a\} \cup V; \hspace{0.2in}
E' = \{(a, v): v \in V\} \\
P  = \{(u, v): (u, v) \in E\}; \hspace{0.2in}
\end{eqnarray*}
We now prove the following lemma which establishes the hardness of the $k$-subset strong rainbow connectivity problem.

\begin{lemma}
\label{lem:vc-subset-equiv}
The graph $G = (V, E)$ is vertex colorable using $k ( \ge 3)$ colors iff the graph $G' = (V', E')$ can be
colored using $k$ colors such that for every pair $(u,v) \in P$ there is a geodesic rainbow path
between $u$ and $v$.
\end{lemma}
\begin{proof}
Assume that $G$ can be vertex colored using $k$ colors; we now show an assignment of colors to the edges
of the graph $G'$. Let $c$ be the color assigned to a vertex $v \in V$; we assign the color $c$ to
the edge $(a, v) \in E'$. Now consider any pair $(u,v)  \in P$. Recall that $(u,v) \in P$ because
there exists an edge $(u,v) \in E$. Since the coloring was a proper vertex coloring, the edges $(a, u)$
and $(a, v)$ in $G'$ are assigned different colors by our coloring. Thus, the path $\{(u,a), (a,v)\}$ is a rainbow
path; further since that is the only path between $u$ and $v$ it is also a strong rainbow path.

To prove the other direction, assume that there exists an edge coloring of $G'$ using $k$ colors
such that between every pair of vertices in $P$ there is a rainbow path. It is easy to see that
if we assign the color $c$ of the edge $(a, v) \in E'$ to the vertex $v \in V$, we get a coloring
that is a proper vertex coloring for $G$.
\qed
\end{proof}
Recall the problem of subset rainbow connectivity where instead of asking for a
geodesic rainbow path between every pair in $P$, we are content with any rainbow
path. Note that our graph $G'$ construed in the above reduction is a tree, in fact
a star and hence between every pair of vertices there is exactly one path. Thus,
all the above arguments apply for the $k$-subset rainbow connectivity problem as well.
As a consequence we can conclude the following:

\begin{lemma}
\label{lem:hardness-src}
For every $k \ge 3$, $k$-subset strong rainbow connectivity and
$k$-subset rainbow connectivity is NP-hard even when the input graph $G$ is a star.
\end{lemma}


\noindent From Lemma \ref{lem:vc-subset-equiv}, 
it can be observed that if subset (strong) rainbow connectivity 
problem had a $f(n)$-factor approximation algorithm then even the 
classical vertex coloring has a $f(n)$-approximation algorithm. 
But we know from \cite{feige1996zero}, that there is no 
$n^{1-\epsilon}$ factor approximation algorithm for vertex 
coloring unless NP=ZPP. Hence subset (strong) rainbow connectivity 
problem is inapproximable within a factor of $n^{1-\epsilon}$ unless NP=ZPP.


\subsection{Hardness of $k$-strong rainbow connectivity}
In this section we prove that the problem of deciding whether a graph $G$ can be strongly
rainbow colored using $k$ colors is polynomial time equivalent to the $k$-subset strong rainbow connectivity problem \textit{i.e.,} the existence of a polynomial time algorithm for one of them will imply the existence of polynomial time algorithm for the other. In particular we prove the lemma.

\begin{lemma}
\label{lem:subset-src-equiv}
The following two problems are reducible to each other in polynomial time:\\
(1) Given a graph $G = (V, E)$ and an integer $k$, decide whether the edges of $G$ 
can be colored using $k$ colors such that between every pair of vertices in $G$
there is a geodesic rainbow path. \\
(2) Given a graph $G = (V, E)$, an integer $k$ and a set of pairs $P \subseteq V \times V$,
decide whether the edges of $G$ can be colored using $k$ colors such that between every pair
$(u,v) \in P$ there is a geodesic rainbow path.
\end{lemma}
\begin{proof}
It suffices to prove that problem (2) reduces to problem (1). Let $\langle G = (V, E), P\rangle$ 
be an instance of the $k$-subset strong rainbow connectivity problem. Using Lemma~\ref{lem:hardness-src},
 we know that
$k$-subset strong rainbow connectivity is NP-hard even when $G$ is a star as well
as the pairs $(u,v) \in P$ are such that both $u$ and $v$ are leaf nodes of the star.
We assume both these properties on the input $\langle G, P\rangle$ and use it crucially in our reduction.
Let us denote the central vertex of the star $G$ by $a$ and the leaf vertices by $L=\{v_1, \ldots ,v_n\}$, that is,
$V = \{a\} \cup L$. Using the graph $G$ and the pairs $P$,
we construct the new graph $G' = (V', E')$ as follows: for every leaf node $v_i \in L$, we introduce two new
vertices $u_i$ and $u_i'$. For every pair of leaf nodes $(v_i, v_j) \in (L \times L) \setminus P$, we introduce two new vertices 
$w_{i,j}$ and $w_{i,j}'$. 

\begin{eqnarray*}
V' &=& V \cup V_1 \cup V_2 \\
V_1 &=& \{u_i: i \in \{1, \ldots, n\} \} \cup \{w_{i,j}: (v_i,v_j) \in (L \times L) \setminus  P\} \ \ \ \ \ \ \ \ \ \ \ \ \ \ \ \ \ \ \ \ \ \ \  \\
V_2 &=& \{u'_i: i \in \{1, \ldots, n\} \} \cup \{w_{i,j}': (v_i,v_j) \in (L \times L) \setminus  P\}
\end{eqnarray*}
The edge set  $E'$ is be defined as follows:
\begin{eqnarray*}
E' &=& E \cup E_1 \cup E_2 \cup E_3\\
E_1 &=& \{ (v_i, u_i): v_i \in L, u_i \in V_1\} \cup \{(v_i, w_{i,j}), (v_j, w_{i,j}):  (v_i,v_j) \in (L \times L) \setminus  P\} \\
E_2 &=& \{ (x, x'): x \in V_1, x' \in V_2\} \\
E_3 &=& \{ (a, x'): x' \in V_2 \}
\end{eqnarray*}
We now prove that $G'$ is strongly rainbow colorable using $k$ colors iff there is an edge coloring of $G$ (using $k$ colors)
such that every pair $(u,v) \in P$ has a strong rainbow path. To prove one direction, we first note that,
for all pairs $(u,v) \in P$, there is a two length path $(u, a, v)$ in $G$ and this path is also present in $G'$.
Further, this path is the only two length path in $G'$ between $u$ and $v$; hence any strong rainbow
coloring of $G'$ using $k$ colors must make this path a rainbow path. This implies that if $G$ cannot be edge
colored with $k$ colors such that every pair in $P$ is strongly rainbow connected, the graph $G'$ cannot be strongly
rainbow colored using $k$ colors.

To prove the other direction, assume that there is a coloring $\coloring: E \rightarrow \{c_1, c_2, \ldots, c_k\}$ 
of $G$ such that all pairs
in $P$ are strongly rainbow connected. We extend this coloring to a strong rainbow coloring of $G'$
as follows:
\begin{itemize}
\item We retain the color on the edges of $G$, i.e. $\coloring'(e) = \coloring(e): e \in E$. 
\item For each edge $(v_i, u_i) \in E_1$, we set $\coloring'(v_i, u_i) = c_3$.
\item For each pair of edges  $\{(v_i, w_{i,j}), (v_j, w_{i,j})\} \in E_1$, we set $\coloring'(v_i, w_{i,j}) = c_1$,
$\coloring'(v_i, w_{i,j}) = c_2$ (Assume without loss of generality that $i < j$). 
\item The edges in $E_2$ form a complete bipartite graph between the vertices in $V_1$ and $V_2$.
To color these edges, we  pick a perfect matching $M$ of size $|V_1|$ and assign $\coloring'(e) = c_1, \forall e \in E_2 \cap M$
and $\coloring'(e) = c_2, \forall e \in E_2 \setminus M$.
\item Finally, for each edge $(a, x') \in E_3$, we set $\coloring'(a, x') = c_3$.
\end{itemize}
It is straightforward to verify that this coloring is indeed a strong rainbow coloring of
the graph $G'$ using no more than $k$ colors. This finishes the proof.
\qed 
\end{proof}

We note that the graph $G'= (V', E')$ constructed in the proof of Lemma~\ref{lem:subset-src-equiv} is in fact
bipartite. The vertex set $V'$ can be partitioned into two sets $A$ and $B$, where $A = \{a\} \cup V_1$
and $B = L \cup V_2$ such that there are not edges between vertices in the same partition.
Thus, using Lemma~\ref{lem:hardness-src} and Lemma~\ref{lem:subset-src-equiv}, 
we conclude the following theorem.

\begin{theorem}
For every $k \ge 3$, deciding whether a given graph $G$ can be strongly rainbow colored using $k$ 
colors is NP-hard. Further, the hardness holds even when the graph $G$ is bipartite.
\end{theorem}


From the same construction, it can be proved that deciding whether a given graph $G$ can be rainbow colored using at most 3 colors is NP-hard.


\section {Rainbow connection number}
\label{rc}
In this section we investigate the complexity of deciding whether the rainbow connection
number of a graph is less than or equal to some 
natural number $k$. We prove the following problem is NP-hard: given a graph $G$
and an integer $k$ decide whether $rc(G) \leq  k$. 
We recall from Lemma~\ref{lem:hardness-src} that the $k$-subset rainbow connectivity problem
is NP-hard.
We reduce from the $k$-subset rainbow connectivity to our problem of deciding
whether a given graph can be rainbow colored using at most $k$ colors.

\subsection{Hardness of $k$-rainbow connectivity}
We first describe the outline of our reduction.
Let $\langle G=(V,E), P \rangle$ be 
the input to the $k$-subset rainbow connectivity problem where $P \subseteq V \times V$. We construct 
a graph $G'=(V',E')$ such that $G$ is a subgraph of $G'$ and 
$rc(G') \leq k$ iff $G$ is $k$-subset rainbow connected.
To construct $G'$, we first 
construct a graph $H_k=(W_k,E_k)$ such that $V \subset W_k$ and $V'=W_k$. Corresponding to the set $P$ of pairs of vertices we 
associate a set of pairs of vertices $P_k$ with respect to the graph $H_k$ (The set $P_k$ is only a relabelling of pairs of vertices in the set $P$). 
Then, the edge disjoint union of the two graphs $G$ and 
$H_k$ will yield the graph $G'$. The graph $G'$ is constructed such that it satisfies two properties:
\begin{enumerate}
\item There exists an edge coloring $\chi$ of $H_k$ with 
$k$ colors such that all pairs of vertices in $G'$ except those 
in $P$ \textit{i.e.,} all pairs in $(V' \times V') \setminus P$ 
are rainbow connected. Thus, if $G$ is $k$-subset rainbow connected 
then there exists $k$-edge coloring of $G$ such that all pairs in 
$P$ will be rainbow connected. From this we prove that $G'$ 
can be rainbow colored using $k$ colors if $G$ is $k$-subset rainbow connected.
\item All paths of length $k$ or less between any pair of vertices in $P$ are contained entirely in $G$ (as a subgraph of $G'$) itself. This ensures that for a rainbow coloring of $G'$ with $k$ colors, any pair of vertices in $P$ should have all its rainbow paths inside $G$ itself. Hence, if $G'$ can be rainbow colored with $k$ colors then $G$ can be edge colored with $k$ colors such that it is subset rainbow connected with respect to $P$.
\end{enumerate}  
We construct the family of graphs $H_k$ inductively. For the base cases $k = 2$ and $k = 3$
we give explicit constructions and prove the properties mentioned above are true for the base
cases. We then show our inductive step and prove that the properties hold. Finally, we describe our
graph $G'$ and prove the correctness of the reduction.



\subsubsection{Construction of $H_2$:}
The construction of the graph $H_2$ is derived from the reduction of
Chakraborty et al.~\cite{chakraborty2008hardness} used to prove that the 2-subset rainbow
connectivity problem is NP-hard.
Let $H_2=(W_2,E_2)$ where the vertex set $W_2$ is defined as follows:
\begin{eqnarray*}
W_2 &=& W_2^{(1)} \cup W_2^{(2)} \\
W_2^{(1)} &=& \{v_{i,2} : i \in \{1, \ldots ,n\} \} \\
W_2^{(2)} &=& \{u_i : i \in \{1, \ldots ,n\}\} \cup \{w_{i,j} : (v_i,v_j) \in (V \times V) \setminus P\} \\
\end{eqnarray*}
The edge set $E_2$ is defined as:
\begin{eqnarray*}
E_2 &=& E_2^{(1)} \cup E_2^{(2)} \cup E_2^{(3)} \\
E_2^{(1)} &=& \{(v_{i,2},u_{i}) : i \in \{1, \ldots ,n\} \} \\
E_2^{(2)} &=& \{(v_{i,2},w_{i,j}),(v_{j,2},w_{i,j}) : (v_i,v_j) \in (V \times V) \setminus P\} \ \ \ \ \   \\
E_2^{(3)} &=& \{(x,y) : x,y \in W_2^{(2)}\}
\end{eqnarray*}

The set of vertices in $W_2^{(1)}$ are referred to as base vertices of $H_2$. 
Let $P_2=\{(v_{i,2},v_{j,2}) : (v_i,v_j) \in P\}$. The graph $H_2$ has 
the property that for all $(v_{i,2},v_{j,2}) \in P_2$ 
there is no path of length $\leq$ 2 between $v_{i,2}$ and $v_{j,2}$. 
Also, if $(v_{i,2},v_{j,2}) \notin P_2$ the shortest path between $v_{(i,2)}$ and $v_{(j,2)}$ is of length 2. 


Let the edge coloring $\chi:E_2 \rightarrow \{c_1,c_2\}$ of $H_2$ be defined as follows.
\begin{itemize}
\item If $e \in E_2^{(1)}$ then $\chi(e)=c_2$.
\item For each pair of edges  $\{(v_{i,2}, w_{i,j}), (v_{j,2}, w_{i,j})\} \in E_2^{(2)}$, 
we set $\chi((v_{i,2}, w_{i,j}))=c_1$ and $\chi((v_{j,2}, w_{i,j}))=c_2$ assuming without loss of generality that $i < j$.
\item If $e \in E_2^{(3)}$ then $\chi(e)=c_1$.
\end{itemize}
It can be verified that for every pair of 
vertices $(x,y)$ in $(W_2 \times W_2) \setminus P_2$, there is a rainbow path between $x$ and $y$.\\


\subsubsection{Construction of $H_3$:}
We now describe the construction of the graph $H_3$. Let $H_3=(W_3,E_3)$ be a graph where the
vertex set $W_3$ is defined as follows:
\begin{eqnarray*}
W_3 &=& W_3^{(1)} \cup W_3^{(2)} \cup W_3^{(3)} \\
W_3^{(1)} &=& \{v_{i,3} : i \in \{1, \ldots ,n\} \} \\
W_3^{(2)} &=& \{u_i : v_i \in V \} \cup \{a_{i,j},b_{i,j} : (v_i,v_j) \in (V \times V) \setminus  P\} \\ 
W_3^{(3)} &=& \{u_i' : v_i \in V \} \cup \{a_{i,j}',b_{i,j}' : (v_i,v_j) \in (V \times V) \setminus  P\}
\end{eqnarray*}
The edge set $E_3$ is defined as:
\begin{eqnarray*}
E_3 &=& E_3^{(1)} \cup E_3^{(2)} \cup E_3^{(3)} \cup E_3^{(4)} \\
E_3^{(1)} &=& \{ (v_{i,3}, u_i): v_i \in V\} \\
E_3^{(2)} &=& \{(v_{i,3}, a_{i,j}), (v_{j,3}, b_{i,j}) : (v_i,v_j) \in (V \times V) \setminus  P\} \\
E_3^{(3)} &=& \{ (x, x') :  x \in W_3^{(2)}, x' \in W_3^{(3)}\} \\
E_3^{(4)} &=& \{ (a_{i,j},b_{i,j}) : (v_i,v_j) \in (V \times V) \setminus  P \}
\end{eqnarray*}

The set of vertices in $W_3^{(1)}$ are referred to as base vertices of $H_3$. 
Let $P_3=\{(v_{i,3},v_{j,3}) : (v_i,v_j) \in P\}$. The graph $H_3$ has 
the property that for all $(v_{i,3},v_{j,3}) \in P_3$ there is no path 
of length $\leq$ 3 between $v_{i,3}$ and $v_{j,3}$. Also, if 
$(v_{i,3},v_{j,3}) \notin P_3$ then the shortest path between $v_{i,3}$ and $v_{j,3}$ is of length 3. \\
Define an edge coloring $\chi:E_3 \rightarrow \{c_1,c_2,c_3\}$ of $H_3$ as follows.
\begin{itemize}
\item For each edge $e=(v_{i,3}, u_i) \in E_3^{(1)}$, we set $\coloring(e) = c_3$.
\item For each pair of edges  $\{(v_{i,3}, a_{i,j}), (v_{j,3}, b_{i,j})\} \in E_3^{(2)}$, 
we set $\coloring((v_{i,3}, a_{i,j})) = c_1$,
$\coloring((v_{j,3}, b_{i,j})) = c_2$ assuming without loss of generality that $i < j$.
\item The edges in $E_3^{(3)}$ form a complete bipartite graph between the vertices in $W_3^{(2)}$ and $W_3^{(3)}$.
To color these edges, we  pick a perfect matching $M$ of size $|W_3^{(2)}|$ and assign $\coloring(e) = c_1, \forall e \in E_3^{(3)} \cap M$
and $\coloring(e) = c_2, \forall e \in E_3^{(3)} \setminus M$.
\item For an edge $e \in E_3^{(4)}$, assign $\chi(e)=c_3$.
\end{itemize}
It can be verified that for every pair of vertices $(x,y)$ in $(W_3 \times W_3) \setminus P$, 
there is a rainbow path between $x$ and $y$.

\subsubsection{Inductive step:}
Assume we have constructed $H_{l-2}=(W_{l-2},E_{l-2})$. 
From $H_{l-2}$ we first construct an intermediate 
graph $\mathcal{H}^{'}_{l-2}$ from which we construct $H_l$. 
The graph $\mathcal{H}^{'}_{l-2}$ is constructed as follows.\\
Construction of $\mathcal{H}^{'}_{l-2}$: Each base vertex $v_{i,l-2}$ in $W_{l-2}$ is split into three vertices $v_{i,l-2}^{(1)}$, $v_{i,l-2}^{(2)}$, $v_{i,l-2}^{(3)}$ and edges are added between them \textit{i.e.,} the vertices $v_{i,l-2}^{(1)}$, $v_{i,l-2}^{(2)}$, $v_{i,l-2}^{(3)}$ form a triangle. Any edge of the form $(w,v_{i,l-2})$ is replaced by three edges $(w,v_{i,l-2}^{(1)}),(w,v_{i,l-2}^{(2)}),(w,v_{i,l-2}^{(3)})$. Formally the graph $\mathcal{H}^{'}_{l-2}=(\mathcal{W}^{'}_{l-2},\mathcal{E}^{'}_{l-2})$ is defined as follows. The set of vertices is $\mathcal{W}^{'}_{l-2}=\mathcal{W}_{l-2}^{(1)} \cup \mathcal{W}_{l-2}^{(2)}$ where,
\begin{eqnarray*}
\mathcal{W}_{l-2}^{(1)} &=& \{ v_{i,l-2}^{(1)}, v_{i,l-2}^{(2)}, v_{i,l-2}^{(3)}\ :\ i \in \{1, \ldots ,n\} \} \\
\mathcal{W}_{l-2}^{(2)} &=& W_{l-2} \backslash \{v_{i,l-2}\ :\ i \in \{1, \ldots ,n\} \} \\
\end{eqnarray*}
and the edge set is $\mathcal{E}^{'}_{l-2} = \mathcal{E}_{l-2}^{(1)} \cup \mathcal{E}_{l-2}^{(2)} \cup \mathcal{E}_{l-2}^{(3)}$ where,
\begin{eqnarray*}
\mathcal{E}_{l-2}^{(1)} &=& \{(v_{i,l-2}^{(j_1)},w)\ :\ (v_{i,l-2},w) \in E_{l-2}, j_1 \in \{1,2,3\} \} \\
\mathcal{E}_{l-2}^{(2)} &=& \{(v_{i,l-2}^{(j_1)},v_{i,l-2}^{(j_2)})\ :\ j_1,j_2 \in \{1,2,3\} \} \\
\mathcal{E}_{l-2}^{(3)} &=& E_{l-2} \backslash \{(v_{i,l-2},w)\ :\ w \in W_{l-2} \} \\
\end{eqnarray*}

\begin{figure}[h!]
\begin{center}
\includegraphics[scale=0.5]{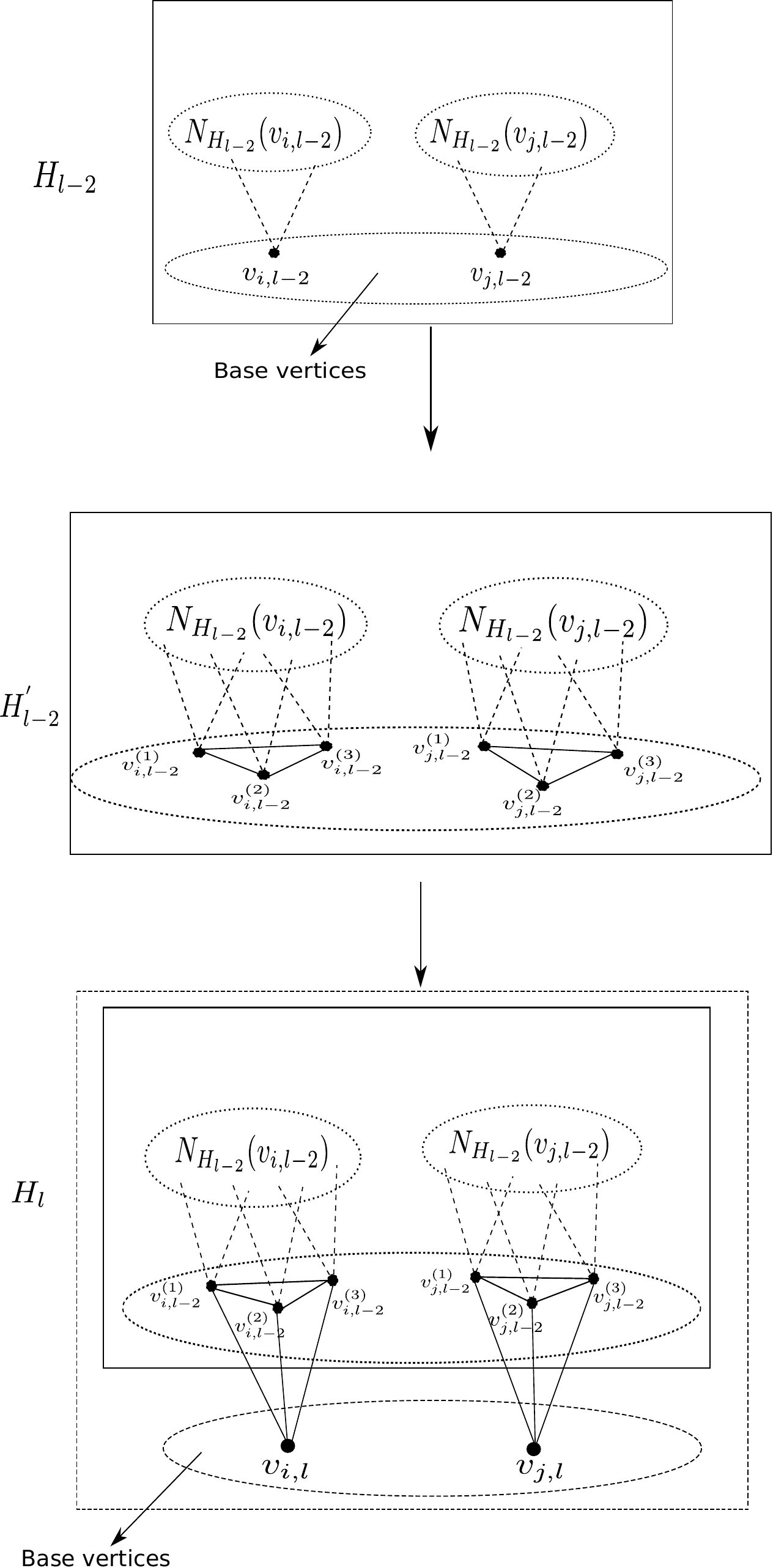}
\end{center}
\caption{Construction of $H_{l}$.}
\end{figure}

\noindent The graph $H_l=(W_l,E_l)$ where $W_{l}=\mathcal{W}^{'}_{l-2} \cup V_l$ and\\ $E_{l}=\mathcal{E}^{'}_{l-2} \cup E'$ such that
\begin{eqnarray*}
V_l &=& \{ v_{1,l}, \ldots ,v_{n,l}\} \\
E' &=& \{ (v_{i,l-2}^{(1)},v_{i,l}),(v_{i,l-2}^{(2)},v_{i,l}),(v_{i,l-2}^{(3)},v_{i,l})\ :\ 1 \leq i \leq n\}
\end{eqnarray*}

The vertices in $V_l$ are the base vertices of $H_l$. 
We now prove that the graph $H_l$ satisfies the  two properties 
stated at the beginning of the section. Before that, 
we define $P_l$ to be the set $\{(v_{i_1,l},v_{i_2,l}) : (v_{i_1},v_{i_2}) \in P\}$ 
which is a set of pairs of vertices in $H_l$. 

\begin{lemma}
\label{pairspath}
If $(v_{i_1,l},v_{i_2,l}) \in P_l$ then there is no path between $v_{i_1,l}$ and $v_{i_2,l}$ of length less than $l$ in $G'$.
\end{lemma}
\begin{proof}
It was shown that the assertion was true for $H_2$ and $H_3$. 
Assume that the assertion is true for $H_{l-2}$. 
Let $(v_{i_1},v_{i_2}) \in P$. Then, $(v_{i_1,l-2},v_{i_2,l-2}) \in P_{l-2}$ 
and hence by induction, every path from $v_{i_1,l-2}$ to $v_{i_2,l-2}$ in $H_{l-2}$
is of length at most $l-2$. It can be seen that 
by the construction of $\mathcal{H}^{'}_{l-2}$ that we are 
not shortening the paths between any two vertices and all 
paths from $v_{i_1,l-2}^{(j_1)}$ to $v_{i_2,l-2}^{(j_2)}$ will still be of length at most 
$l-2$ for $j_1,j_2 \in \{1,2,3\}$. Consider the graph $H_l$. 
Since, $N_{H_l}(v_{i_1,l})=\{v_{i_1,l-2}^{(1)},v_{i_1,l-2}^{(2)},v_{i_1,l-2}^{(3)}\}$ 
and $N_{H_l}(v_{i_2,l})=\{v_{i_2,l-2}^{(1)},v_{i_2,l-2}^{(2)},v_{i_2,l-2}^{(3)}\}$ 
(here, $N_G(v)$ denotes the neighbourhood of vertex $v$ in graph $G$), 
there cannot be path of length less than or equal to $l$ between $v_{i_1,l}$ and $v_{i_2,l}$ in $H_l$. 
\qed
\end{proof}

\begin{lemma}
\label{nonpairspath}
If $(v_{i_1,l},v_{i_2,l}) \notin P_l$ then the shortest path between $v_{i_1,l}$ and $v_{i_2,l}$ is of length $l$. 
\end{lemma} 
\begin{proof}
It was shown that the assertion was true for $H_2$ and $H_3$. Assume that the assertion is true for $H_{l-2}$. Let $v_{i_1},v_{i_2} \notin P$. Then, $(v_{i_1,l-2},v_{i_2,l-2}) \in P_{l-2}$ and hence by induction, the shortest path from $v_{i_1,l-2}$ to $v_{i_2,l-2}$ is of length $l-2$. The construction of $\mathcal{H}^{'}_{l-2}$ does not shorten any path between any two vertices and hence the shortest path from $v_{i_1,l-2}^{(j_1)}$ to $v_{i_2,l-2}^{(j_2)}$ will be of length exactly 
$l-2$ for $j_1,j_2 \in \{1,2,3\}$. Now consider the graph $H_l$. 
Since, $N_{H_l}(v_{i_1,l})=\{v_{i_1,l-2}^{(1)},v_{i_1,l-2}^{(2)},v_{i_1,l-2}^{(3)}\}$ 
and $N_{H_l}(v_{i_2,l})=\{v_{i_2,l-2}^{(1)},v_{i_2,l-2}^{(2)},v_{i_2,l-2}^{(3)}\}$, the shortest 
path $v_{i_1,l}$ and $v_{i_2,l}$ in $H_l$ is of length $l$.
\qed
\end{proof}

\begin{lemma}
\label{nonpairscolor}
There exists an edge coloring of $H_l$ with $l$ colors such that all pairs in $(W_l \times W_l) \setminus P_l$ are rainbow connected.
\end{lemma}
\begin{proof}
It was shown that the assertion is true for $H_2,H_3$. Assume that the assertion is true for $H_{l-2}$. Let $\chi: E_{l-2} \rightarrow \{c_1, \ldots ,c_{l-2}\}$ be an edge coloring of $H_{l-2}$ such that all pairs in $(W_{l-2} \times W_{l-2}) \backslash P_{l-2}$ are rainbow connected. Using $\chi$, we define an edge coloring $\chi': E_l \rightarrow \{c_1, \ldots ,c_l\}$ of $H_{l}$ as follows.\\
\indent $-$ If $e \in \mathcal{E}_{l-2}^{(1)}$, then if $e=(v_{i,l-2}^{(1)},w)$ or $e=(v_{i,l-2}^{(2)},w)$ then $\chi'(e)=\chi((v_{i,l-2},w))$ or else if $e=(v_{i,l-2}^{(3)},w)$ then $\chi'(e)=c_{l-1}$.\\
\indent $-$ If $e \in \mathcal{E}_{l-2}^{(2)}$, then $\chi'(e)=c_{l}$.\\
\indent $-$ If $e \in \mathcal{E}_{l-2}^{(3)}$, then $\chi'(e)=\chi(e)$.\\
\indent $-$ If $e \in E'$, then if $e=(v_{i,l-2}^{(1)},v_{i,l})$ then $\chi'(e)=c_{l-1}$ or else if $e=(v_{i,l-2}^{(2)},v_{i,l})$ or $e=(v_{i,l}^{(3)},v_{i,l})$ then $\chi'(e)=c_{l}$.\\

\begin{figure}[h!]
\begin{center}
\includegraphics[scale=0.5]{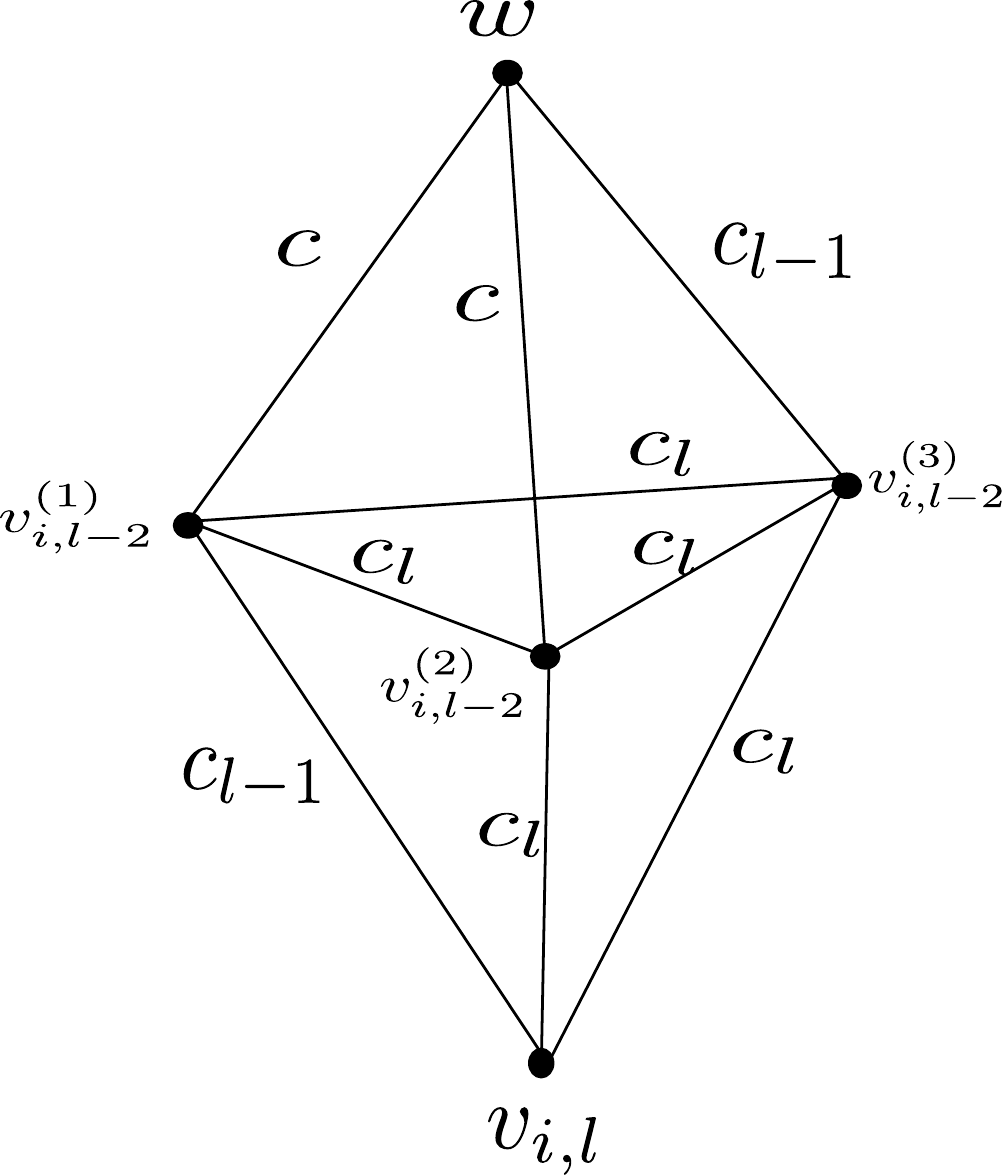}
\end{center}
\caption{Edge coloring of $H_l$: $c$ is the color originally used to color the edge $(v_{i,l-2},w)$.} 
\end{figure}

We now show that the $l$-edge coloring $\chi'$ makes $H_l$ rainbow connected. \\

\noindent \textbf{Claim.} Let $v,w \in W_{l}$. If $(v,w) \notin P_l$, then $v$ and $w$ are rainbow connected in $H_{l}$.\\
\textit{Proof.} We will consider the following cases.

\noindent \textbf{Case (i)} $v,w \in \mathcal{W}_{l-2}^{(2)}$: Since $\mathcal{W}_{l-2}^{(2)} \subset W_{l-2}$, there is a rainbow path $\mathcal{P}$ from $v$ to $w$ in $W_{l-2}$. If $\mathcal{P}$ consisted of edges only from $\mathcal{E}_{l-2}^{(3)}$ then these edges are also in $E_{l-2}$ and by our coloring scheme, the rainbow path $\mathcal{P}$ is also present in $H_l$. So assume that $\mathcal{P}$ contains base vertices \textit{i.e.,} let $\mathcal{P}=v \rightsquigarrow v_{i_1,l-2} \cdots \rightsquigarrow v_{i_s,l-2} \rightsquigarrow w$ where the distinct $i_1, \ldots ,i_s$ belong to $\{1, \ldots ,n\}$. It can be observed that the path $\mathcal{P}'=v \rightsquigarrow v_{i_1,l-2}^{(1)} \cdots \rightsquigarrow v_{i_s,l-2}^{(1)} \rightsquigarrow w$ is rainbow path from $v$ to $w$ in $W_{l-2}$.  

\noindent \textbf{Case (ii)} $v \in \mathcal{W}_{l-2}^{(1)}, w \in \mathcal{W}_{l-2}^{(2)}$: Let $v=v_{i,l-2}^{(j_1)}$ for $i \in \{1, \ldots ,n\},j_1 \in \{1,2,3\}$. There is a rainbow path in $H_{l-2}$ from $v_{i,l-2}$ to $w$. From an argument similar to the previous case, there exists a corresponding rainbow path from $v=v_{i,l-2}^{(1)}$ to $w$ denoted by $v_{i,l-2}^{(1)} \overset{R_{H_l}}{\rightsquigarrow} w$ which does not use the colors $c_{l-1}$ and $c_l$. Similarly there exists a rainbow path from $v_{i,l-2}^{(2)}$ to $w$. Also, the path $v_{i,l-2}^{(3)} - v_{i,l-2}^{(1)} \overset{R_{H_l}}{\rightsquigarrow} w$ from $v_{i,l-2}^{(3)}$ to $w$ is a rainbow path.     

\noindent \textbf{Case (iii)} $v,w \in \mathcal{W}_{l-2}^{(1)}$: Let $v=v_{(i_1,l-2)}^{(j_1)}$, $w=v_{(i_2,l-2)}^{(j_2)}$ for $j_1,j_2 \in \{1,2,3\}$. And let $w_1 \in N_{H_{l-2}}(v_{i,l-2})$, $w_2 \in N_{H_{l-2}}(v_{j,l-2})$. From Case (ii), there are rainbow paths from $w_1$ to $v_{j,l-2}^{(1)}$ (denoted by $w_1 \overset{R_G}{\rightsquigarrow} v_{j,l-2}^{(1)}$) and from $w_2$ to $v_{i,l-2}^{(1)}$ (denoted by $w_2 \overset{R_G}{\rightsquigarrow} v_{i,l-2}^{(2)}$). The following table shows the rainbow paths for all possible cases of $v$ and $w$:\\

\begin{center}
\scalebox{0.7}{

\begin{tabular}{ | c || c | c | c |}
\hline
  \  & $v_{j,l-2}^{(1)}$  & $v_{j,l-2}^{(2)}$ & $v_{j,l-2}^{(3)}$  \\ \hline \hline
   $v_{i,l-2}^{(1)}$ & $v_{i,l-2}^{(1)} \overset{R_{H_l}}{\rightsquigarrow} w_2 - v_{j,l-2}^{(3)} - v_{j,l-2}^{(1)}$ &  $v_{i,l-2}^{(1)} \overset{R_{H_l}}{\rightsquigarrow} w_2 - v_{j,l-2}^{(3)} - v_{j,l-2}^{(2)}$ & $v_{i,l-2}^{(1)} \overset{R_{H_l}}{\rightsquigarrow} w_2 - v_{j,l-2}^{(3)}$ \\ \hline
     $v_{i,l-2}^{(2)}$ & $-$  & $v_{i,l-2}^{(2)} \overset{R_{H_l}}{\rightsquigarrow} w_2 - v_{j,l-2}^{(3)} - v_{j,l-2}^{(2)}$ & $v_{i,l-2}^{(2)} \overset{R_{H_l}}{\rightsquigarrow} w_2 - v_{j,l-2}^{(3)}$  \\ \hline
       $v_{i,l-2}^{(3)}$ & $-$  & $-$ & $v_{i,l-2}^{(3)} - w_1 \overset{R_{H_l}}{\rightsquigarrow} v_{i,l-2}^{(1)} - v_{j,l-2}^{(3)} $  \\
       \hline
       \end{tabular}
}
\end{center}

\ \\
\ \\
       \noindent \textbf{Case (iv)} $v,w \in V_l$ and $(v,w) \notin P_l$: Let $v=v_{i,l},w=v_{i,l}$. Since $(v_{i,l},v_{j,l}) \notin P_l$, this means that $(v_{i,l-2},v_{j,l-2}) \notin P_{l-2}$. Hence, there exists a rainbow path from $v_{i,l-2}$ to $v_{j,l-2}$ in $H_{l-2}$. Correspondingly in the graph $\mathcal{H}'_{l-2}$, there exists a rainbow path from $v_{i,l-2}^{(1)}$ to $v_{j,l-2}^{(2)}$ denoted by $v_{i,l-2}^{(1)} \rightsquigarrow v_{j,l-2}^{(2)}$. It can be observed that the edge colored path $v_{i,l-2} - v_{i,l-2}^{(1)} \rightsquigarrow v_{j,l-2}^{(2)} - v_{j,l-2}$ is a rainbow path. 

       \noindent \textbf{Case (v)} $v \in W_{l-2}, w \in V_l$: Let $w=v_{i,l} \in V_l$. There exists a rainbow path from $v_{i,l-2}$ to $v$ in $H_{l-2}$ . Correspondingly there exists a rainbow path from $v_{i,l-2}^{(1)}$ to $v$ in $H_l$, denoted by $v_{i,l-2}^{(1)} \rightsquigarrow v$, which does not use colors $c_{l-1}$ and $c_l$. Hence, the path $v_{i,l} - v_{i,l-2}^{(1)} \rightsquigarrow v$ is a rainbow path in $H_l$.
\qed
       \end{proof}


\subsubsection{Reduction:}
\noindent  We have gathered all the tools to formulate the reduction. Let the instance for $k$-subset rainbow connectivity be $\langle G=(V,E),P \rangle$. We construct a graph $\langle G'\rangle$ as an instance for $k$-rainbow connectivity as follows:\\
       1) Construct $H_k=(V_k,E_k)$.\\
       2) Let $V'=W_k,\ E'=E_k \cup \{(v_{i,k},v_{j,k})\ :\ (v_i,v_j) \in E\}$.
\par $G'=(V',E')$ is the required graph. We call the induced subgraph of $G'$ containing the base vertices as base graph $G_k$. It can be seen that $G_k$ is isomorphic to $G$.
\par Let $(v_{i,k},v_{j,k}) \in P$. Consider a path $\mathcal{P}=v_{i,k}-x_1 \rightsquigarrow x_s-v_{j,k}$ between $v_{i,k}$ and $v_{j,k}$ of length $k$ or less. All the edges in $\mathcal{P}$ cannot be in $H_k$ from Lemma~\ref{pairspath}. Hence $\mathcal{P}$ contains at least one edge from $G_k$. Consider a subpath $\mathcal{P}'=x \rightsquigarrow y$ of $\mathcal{P}$ such that $x,y$ are base vertices and all the edges in $\mathcal{P}'$ are in $H_k$. If $\mathcal{P}' \neq \emptyset$ then $\mathcal{P} \neq \mathcal{P}'$ since $\mathcal{P}$ contains at least one edge from $G_k$. 
From Lemma~\ref{pairspath} and Lemma~\ref{nonpairspath}, $\mathcal{P}'$ is of length 
at least $k$. Hence, length of $\mathcal{P}$ is definitely more than $k$,
 contradicting the fact that length of $\mathcal{P}$ is $k$ or less.  Hence $\mathcal{P}'=\emptyset$ and all the edges in $\mathcal{P}$ are in $G_k$. Thus, any path of length $k$ or less between any pair of vertices in $P$ is entirely contained in $G_k$ itself.
\par We now prove the main theorem of this section which proves that $k$-rainbow connectivity problem is NP-hard. 

\begin{lemma}
$G$ is $k$-subset rainbow connected iff $G'$ is $k$-rainbow connected.
\end{lemma}
\begin{proof}
       ``If": Let $\chi:G' \rightarrow \{c_1, \ldots ,c_k\}$ be an edge coloring of $G'$ with $k$ colors such that every pair of vertices in $G'$ have a rainbow path between them. We define $\chi'$ to be a $k$-edge coloring of $G$ as follows:  $\chi'(v_i,v_j)=\chi((v_{i,k},v_{j,k}))$ if the edge $(v_{i,k},v_{j,k}) \in G_k$. By the observations made prior to this lemma, all paths between $v_{i,k},v_{j,k}$ of length $k$ or less are in $G_k$ itself. Hence the entire rainbow path between $v_{i,k}$ and $v_{j,k}$ must lie in $G_k$ itself. Correspondingly, there is a rainbow path between $v_i$ and $v_j$ in $G$. \\
       ``Only if": If $G$ is $k$-subset rainbow connected then let $\chi$ be a edge coloring of $G$ such that $\chi:E \rightarrow \{c_1, \ldots c_k\}$ makes it subset rainbow connected. Also, let $\chi':W_k \rightarrow \{c_1, \ldots c_k\}$ be an edge coloring of $H_k$ such that all pairs of vertices in $(W_k \times W_k) \setminus P_k$ are rainbow connected. We now define an edge coloring $\chi'':E' \rightarrow \{c_1, \ldots ,c_k\}$ in $G'$ using $\chi$ and $\chi'$ as follows. Let $e \in G'$. If $e \in H_k$ then $\chi''(e)=\chi'(e)$ else $\chi''(e)=\chi(e)$. The rainbow paths between every pair of vertices in $P_k$ will be in $G_k$ itself. From Lemma \ref{nonpairscolor}, all the pairs of vertices in $V' \backslash P_k$ will be rainbow connected. Since there is a rainbow path between every pair of vertices in $G'$, the graph $G'$ is rainbow connected.
\qed
       \end{proof}

\noindent Therefore, we now conclude the following theorem.

\begin{theorem}
For every $k \geq 3$, deciding whether $rc(G) \leq 3$ is NP-hard.
\end{theorem}


\section{Conclusion}
In this paper, we showed that given any natural number $k$ and a graph $G$ it is NP-hard to determine whether $rc(G) \leq k$. We also show given a bipartite graph, the problem of determining whether $src(G) \leq k$ is NP-hard. \\

\section{Acknowledgements}   
We would like to thank Deepak Rajendraprasad and Dr. L. Sunil Chandran for the useful discussions on the topic.

\bibliographystyle{plain}
\bibliography{prabhanjan}

\begin{thebibliography}{1}

\bibitem{basavaraju2010rainbow}
M.~Basavaraju, L.S. Chandran, D.~Rajendraprasad, and A.~Ramaswamy.
\newblock {Rainbow connection number and radius}.
\newblock {\em Arxiv preprint arXiv:1011.0620}, 2010.

\bibitem{caro2008rainbow}
Y.~Caro, A.~Lev, Y.~Roditty, Z.~Tuza, and R.~Yuster.
\newblock {On rainbow connection}.
\newblock {\em the electronic journal of combinatorics}, 15(R57):1, 2008.

\bibitem{chakraborty2008hardness}
S.~Chakraborty, E.~Fischer, A.~Matsliah, and R.~Yuster.
\newblock {Hardness and Algorithms for Rainbow Connection}.
\newblock {\em Arxiv preprint arXiv:0809.2493}, 2008.

\bibitem{chandran2010rainbow}
L.S. Chandran, A.~Das, D.~Rajendraprasad, and N.M. Varma.
\newblock {Rainbow connection number and connected dominating sets}.
\newblock {\em Arxiv preprint arXiv:1010.2296}, 2010.

\bibitem{chartrand2008rainbow}
G.~Chartrand, G.L. Johns, K.A. McKeon, and P.~Zhang.
\newblock {Rainbow connection in graphs}.
\newblock {\em Math. Bohem}, 133(1):85--98, 2008.

\bibitem{chen2011complexity}
L.~Chen, X.~Li, and Y.~Shi.
\newblock {The complexity of determining the rainbow vertex-connection of
  graphs}.
\newblock {\em Arxiv preprint arXiv:1101.3126}, 2011.

\bibitem{feige1996zero}
U.~Feige and J.~Kilian.
\newblock {Zero knowledge and the chromatic number}.
\newblock {\em ccc}, page 278, 1996.

\bibitem{Krivelevich:2010:RCG:1753069.1753071}
Michael Krivelevich and Raphael Yuster.
\newblock The rainbow connection of a graph is (at most) reciprocal to its
  minimum degree.
\newblock {\em J. Graph Theory}, 63:185--191, March 2010.

\bibitem{li2011rainbow}
H.~Li, X.~Li, and S.~Liu.
\newblock {Rainbow connection in graphs with diameter 2}.
\newblock {\em Arxiv preprint arXiv:1101.2765}, 2011.

\end{thebibliography}
\end{document}